\documentclass[11pt,a4paper]{amsart}
\usepackage[foot]{amsaddr}
\usepackage{ifxetex}
\ifxetex
  \usepackage[no-math]{fontspec}
\else
\fi
\usepackage{amsmath}
\usepackage{amsfonts}
\usepackage{amssymb}
\usepackage{amsthm}
\usepackage{lmodern}
\usepackage{fullpage}
\usepackage{microtype}
\ifxetex 
  \usepackage[libertine]{newtxmath}
\else
  \usepackage{newtxmath}
\fi
\usepackage[tt=false]{libertine} 
\usepackage{caption}
\usepackage{tcolorbox}
\usepackage{bbm}
\usepackage{hyperref, color}
\hypersetup{colorlinks=true,citecolor=blue, linkcolor=blue, urlcolor=blue}
\usepackage[linesnumbered,boxed,ruled,vlined]{algorithm2e}
\usepackage{bm}
\usepackage{bbm}
\usepackage[numbers]{natbib}
\usepackage{xcolor}
\usepackage{enumerate} 
\usepackage{enumitem}
\usepackage{ragged2e}
\usepackage{tabularx}
\usepackage{array}
\usepackage{longtable}
\usepackage{hhline}
\usepackage{multirow}
\usepackage{cleveref}
\newcolumntype{L}[1]{>{\raggedright\arraybackslash}p{#1}}
\newcolumntype{C}[1]{>{\centering\arraybackslash}m{#1}}
\newcolumntype{R}[1]{>{\raggedleft\arraybackslash}p{#1}}

\usepackage{makecell}
\usepackage{footnote}
\usepackage{float}
\usepackage{xifthen}
\makesavenoteenv{tabular}
\usepackage{todonotes}

\usetikzlibrary{patterns,calc}

\newcommand{\TODO}[1]{\typeout{TODO: \the\inputlineno: #1}\textbf{{\color{red}[[[ #1 ]]]}}}

\newcommand{\sample}{\textnormal{\textsf{Sample}}}
\newcommand{\esample}{\textnormal{\textsf{SampleEdge}}}
\newcommand{\Ind}{\textnormal{\textsf{Ind}}}
\newcommand{\dtv}{d_{\rm TV}}
\newcommand{\abs}[1]{\left\vert#1\right\vert}

\newcommand{\numP}{{\textnormal{\#\textbf{P}}}}

\newcommand{\numSFO}{\textnormal{\#\textsc{SFO}}}

\SetKwRepeat{Do}{do}{while}

\newtheorem{theorem}{Theorem}[section]

\newtheorem{claim}[theorem]{Claim}
\newtheorem*{claim*}{Claim}

\newtheorem{lemma}[theorem]{Lemma}

\theoremstyle{definition}

\newtheorem{definition}[theorem]{Definition}

\newtheorem*{remark*}{Remark}

\renewcommand{\Pr}[2][]{ \ifthenelse{\isempty{#1}}
  {\mathop{\mathbf{Pr}}\left[#2\right]} {\mathop{\mathbf{Pr}}_{#1}\left[#2\right]} }
\newcommand{\E}[2][]{ \ifthenelse{\isempty{#1}}
  {\mathop{\mathbf{E}}\left[#2\right]}
  {\mathop{\mathbf{E}}_{#1}\left[#2\right]} }
\newcommand{\Var}[2][]{ \ifthenelse{\isempty{#1}}
  {\mathbf{\mathbf{Var}}\left[#2\right]}
  {\mathbf{\mathbf{Var}}_{#1}\left[#2\right]} }

\def\^#1{\mathbb{#1}} 
\def\*#1{\mathbf{#1}} 
\def\+#1{\mathcal{#1}} 
\def\-#1{\mathrm{#1}} 
\def\=#1{\boldsymbol{#1}} 

\newcommand{\defeq}{\coloneq}
\newcommand{\eps}{\varepsilon}

\usepackage{amssymb}
\usepackage{stackengine}
\usepackage{scalerel}
\usepackage{xcolor}
\usepackage{graphicx}

\newboolean{doubleblind}
\setboolean{doubleblind}{false}

\title{Sink-free orientations: a local sampler with applications}

\ifdoubleblind
\author{Author(s)}
\else

\author{Konrad Anand, Graham Freifeld, Heng Guo, Chunyang Wang, Jiaheng Wang}

\address[Konrad Anand, Graham Freifeld, and Heng Guo]{School of Informatics, University of Edinburgh, Informatics Forum, Edinburgh, EH8 9AB, United Kingdom.}
\address[Chunyang Wang]{State Key Laboratory for Novel Software Technology, New Cornerstone Science Laboratory, Nanjing University, 163 Xianlin Avenue, Nanjing, Jiangsu Province, China.}
\address[Jiaheng Wang]{Faculty of Informatics and Data Science, University of Regensburg, Bajuwarenstra{\ss}e 4, 93053 Regensburg, Germany.}

\thanks{This project has received funding from the European Research Council (ERC) under the European Union's Horizon 2020 research and innovation programme (grant agreement No.~947778). 
Jiaheng Wang also acknowledges support from the ERC (grant agreement No. 101077083).}

\fi

\date{}
    
\begin{document}

\begin{abstract}
  \sloppy
  For sink-free orientations in graphs of minimum degree at least $3$,
  we show that there is 
  a deterministic approximate counting algorithm that runs in time $O((n^{73}/\eps^{72})\log(n/\eps))$,
  a near-linear time sampling algorithm,
  and a randomised approximate counting algorithm that runs in time $O((n/\eps)^2\log(n/\eps))$, where $n$ denotes the number of vertices of the input graph and $0<\eps<1$ is the desired accuracy.
  All three algorithms are based on a local implementation of the sink popping method (Cohn, Pemantle, and Propp, 2002) under the partial rejection sampling framework (Guo, Jerrum, and Liu, 2019). 
\end{abstract}


\allowdisplaybreaks
\maketitle

\section{Introduction}

The significance of counting has been recognised in the theory of computing since the pioneering work of Valiant \cite{Val79,Val79a}.
In the late 80s, a number of landmark approximate counting algorithms \cite{JS89,DFK91,JS93} were discovered.
A common ingredient of these algorithms is the computational equivalence between approximate counting and sampling for self-reducible problems \cite{jerrum1986random}.
The reduction from counting to sampling decomposes the task into a sequence of marginal probability estimations,
each of which is tractable for sampling techniques such as Markov chains.
However, while only the marginal probability of one variable is in question,
simulating Markov chains requires keeping track of the whole state of the instance, and thus is obviously wasteful.
It is more desirable to draw samples while accessing only some local structure of the target variable.
We call such algorithms local samplers.

The first such local sampler was found by Anand and Jerrum \cite{anand2021perfect},
who showed how to efficiently generate perfect local samples for spin systems even when the underlying graph is infinite.
Using local information is essential here as it is not possible to perfectly simulate the whole state.
Subsequently, Feng, Guo, Wang, Wang, and Yin \cite{feng2023towards} found an alternative local sampler, namely the so-called coupling towards the past (CTTP) method, which yields local implementations of rapid mixing Markov chains.
It is also observed that sufficiently efficient local samplers lead to immediate derandomisation via brute-force enumeration.
Moreover, local samplers are crucial to obtain sub-quadratic time approximate counting algorithms for spin systems \cite{AFFGW25}.
Thus, local samplers are highly desirable algorithms as they can lead to fast sampling, fast approximate counting, and deterministic approximate counting algorithms.

Guo, Jerrum, and Liu \cite{GJL19} introduced partial rejection sampling (PRS) as yet another efficient sampling technique.
This method generalises the cycle-popping algorithm for sampling spanning trees \cite{Wil96} and the sink-popping algorithm for sampling sink-free orientations \cite{CPP02}.
It also has close connections with the Lov\'{a}sz local lemma \cite{EL75}.
For extremal instances (in the sense of \cite{KS11}), PRS is just the celebrated Moser-Tardos algorithm for the constructive local lemma \cite{moser2010constructive}.
The most notable application of PRS is the first fully polynomial-time randomised approximation scheme (FPRAS) for all-terminal network reliability \cite{GJ19}. 
On the other hand, it is still open if all-terminal reliability and counting sink-free orientations admit deterministic fully polynomial-time approximation schemes (FPTASes). 
Thus, in view of the aforementioned derandomisation technique \cite{feng2023towards}, 
a local implementation of PRS is a promising way to resolve these open problems.

In this paper, we make some positive progress for sink-free orientations (SFOs).
Given an undirected graph $G=(V,E)$, a sink-free orientation of $G$ is an orientation of edges such that each vertex has at least one outgoing edge. 
SFOs were first studied by Bubley and Dyer \cite{BD97a} as a restricted case of \textsc{Sat}.\footnote{As a side note, we remark that SFOs are also introduced in the context of distributed computing under the name of sinkless orientations, 
where they are used to give a lower bound for the distributed Lov\'{a}sz local lemma \cite{BFHKLRSU16}.}
They showed that exact counting of SFOs is $\numP$-complete, and thus is unlikely to have a polynomial-time algorithm. 
For approximate counting and sampling, in \cite{BD97b}, they showed that a natural Markov chain has an $O(m^3)$ mixing time, where $m$ is the number of edges. 
Later, Cohn, Pemantle, and Propp \cite{CPP02} introduced an exact sampler, namely the aforementioned sink-popping algorithm that runs in $O(nm)$ time in expectation, where $n$ is the number of vertices.
Using the PRS framework, Guo and He \cite{guo2020tight} improved the running time of sink-popping to $O(n^2)$,
and constructed instances where this running time is tight.
It is open whether a faster sampling algorithm or an FPTAS exists.

Our main result is a local sampler based on PRS for SFOs.
Using this local sampler, for graphs of minimum degree $3$,
we obtain a deterministic approximate counting algorithm that runs in time $O((n^{73}/\eps^{72})\log(n/\eps))$,
a near-linear time sampling algorithm,
and a randomised approximate counting algorithm that runs in time $O((n/\eps)^2\log(n/\eps))$,
where $\eps$ is the given accuracy.
All three algorithms appear difficult to obtain using previous techniques.
We will describe the results in more detail in the next section.

\subsection{Our contribution and technique overview}

Our local sampler works for a slight generalisation of SFOs,
which are intermediate problems required by the standard counting to sampling reduction \cite{jerrum1986random}.
In these problems, a subset $S$ of vertices is specified, which are required to be sink-free,
and the task is to estimate the probability of a vertex $v$ not in $S$ not being a sink.

Before describing our technique, let us first review the sink-popping algorithm,
(which is a special case of PRS and the same as the Moser-Tardos algorithm \cite{moser2010constructive} as the instance is extremal).
We orient each edge uniformly at random.
As long as there is a sink in $S$, we select one such vertex, arbitrarily, and rerandomise all edges incident to it, until there is no sink belong to $S$.

Our key observation is that it is unnecessary to simulate all edges to decide if $v$ is a sink.
In particular, if, at any point of the execution of the algorithm, $v$ is a sink,
then no adjacent edges will ever be resampled and $v$ stays a sink till the algorithm finishes.
On the other hand, if at any point $v$ belongs to a cycle, a path leading to a cycle, or a nonempty path leading to some vertex not in $S$, 
then the orientations of all edges involved will not be resampled, and $v$ stays a non-sink until the algorithm terminates.
Thus, this observation gives us an early termination criterion for determining whether $v$ is a sink or not. 
Moreover, since in the sink-popping algorithm, the order of sinks popped can be arbitrary,
we can reveal the random orientation of edges strategically, and pop sinks if necessary.
To be more precise, we first reveal the edges adjacent to $v$ one by one.
Once there is an outgoing edge $(v,u)$, we then move to $u$ and repeat this process.
If any sink is revealed, we erase the orientations of all its adjacent edges and backtrack.
Eventually, one of the two early termination rules above will kick in, and this gives us a local sample.

Ideally, we want our local sampler to run in $O(\log n)$ time, where $n$ is the number of vertices.
Unfortunately, the one described above does not necessarily terminate this fast.
To see this, consider a sequence of degree $2$ vertices, where at each step there is equal probability to move forward or backtrack.
Resolving such a path of length $\ell$ would require $\Theta(\ell^2)$ time.
On the other hand, when the minimum degree of the input graph is at least $3$, the length of the path followed by the sampler forms a submartingale. 
The vertex $v$ can be a sink only if this path has length $0$.
Thus, once the length of the path is at least $C\log n$ for some constant $C$,
the probability of $v$ being a sink is very small.
This allows us to truncate the local sampler with only a small error.

The FPTAS for \numSFO{} follows from the derandomisation method of \cite{feng2023towards} to the truncated local sampler.
By $\varepsilon$-approximation,
we mean an estimate $\widetilde{Z}$ such that $1-\eps\le\frac{\widetilde{Z}}{Z}\le 1+\eps$, where $Z$ is the target quantity.
Also, all our algorithms work for not necessarily simple graphs.

\begin{theorem}[deterministic approximate counting]\label{thm:sfo-deterministic-counting}
  For graphs with minimum degree at least $3$, there exists a deterministic algorithm that, given $0<\eps<1$, outputs an $\varepsilon$-approximation to the number of sink-free orientations with running time $O((n^{73}/\eps^{72})\log(n/\varepsilon))$, where $n$ is the number of vertices. 
\end{theorem}

Although high, the constant exponent in the running time of \Cref{thm:sfo-deterministic-counting} is actually the most interesting feature of our algorithm. 
In contrast, the running time of most known FPTASes \cite{weitz06counting, BG06, BGKNT07, HSV18, Bar16, patel2017deterministic, Moi19, feng2023towards, CFGZZ24} has an exponent that depends on some parameter (such as the maximum degree) of the input graph.
There are exceptions, for example, \cite{LLY13,GL18}, but the exponents of their running times still depend on the parameters of the problem (not of the instance).

For fast sampling, we need a slight modification of the idea above to sample orientations of edges one by one,
resulting in the following approximate sampler. 

\begin{theorem}[fast sampling]\label{thm:sfo-fast-sampling}
  For graphs with minimum degree at least $3$, there exists a sampling algorithm that, given  $0<\varepsilon<1$, outputs a random orientation $\sigma$ such that $\sigma$ is $\varepsilon$-close to a uniform random sink-free orientation in total variation distance, with running time $O\left(m\log\left(\frac{m}{\varepsilon}\right)\right)$, where $m$ is the number of edges. 
\end{theorem}

Our sampler runs in $\widetilde{O}(m)$ time\footnote{The $\widetilde{O}$ notation hides logarithmic factors.} instead of the $O(n^2)$ time that sink-popping requires, at the cost of generating an approximate sample instead of a perfect sample. 
This improves over sink-popping when $m=o(n^2/\log n)$ and leads to a faster FPRAS using the counting-to-sampling reduction \cite{jerrum1986random}.
In fact, the running time of the FPRAS can be improved further by directly invoking the truncated local sampler in the reduction.

\begin{theorem}[fast approximate counting]\label{thm:sfo-fast-counting}
  For graphs with minimum degree at least $3$, there exists a (randomised) algorithm that, given $0<\eps<1$, outputs a quantity that is an $\varepsilon$-approximation with probability at least $3/4$ to the number of sink-free orientations.
  The running time is $O((n/\eps)^2\log(n/\eps))$, where $n$ is the number of vertices. 
\end{theorem}

The success probability $3/4$ in \Cref{thm:sfo-fast-counting} is standard in the definition of FPRAS,
and can be easily amplified by taking the median of repeated trials and applying the Chernoff bound.

Note that directly combining \Cref{thm:sfo-fast-sampling} with the counting-to-sampling reduction results in an $\widetilde{O}(n m/\eps^2)$ running time.
\Cref{thm:sfo-fast-counting} is faster when $m=\omega(n)$.
Previously, the best running time for approximate counting is $\widetilde{O}(n^3/\eps^2)$,
via combining the $O(n^2)$ time sink-popping algorithm \cite{guo2020tight} with simulated annealing (see, for example, \cite[Lemma 12]{guo2020tight}).
\Cref{thm:sfo-fast-counting} improves over this by roughly a factor of $n$. 
In very dense graphs (when $m=\Omega(n^2)$),
\Cref{thm:sfo-fast-counting} achieves near-linear time, which appears to be rare for approximate counting. 

There are a plethora of fast sampling and deterministic approximate counting techniques by now.
However, it appears difficult to achieve our results without the new local sampler.
For example, the coupling of Bubley and Dyer \cite{BD97b} does not seem to improve with the minimum degree requirement.
On a similar note, the recent deterministic counting technique of \cite{CFGZZ24} requires a distance-decreasing Markov chain coupling, whereas the Bubley-Dyer coupling is distance non-increasing.
In any case, even if the technique of \cite{CFGZZ24} applied, it would not imply a running time with a constant exponent.
Other fast sampling and FPTAS techniques, such as spectral independence \cite{anari2020spectral,chen2021optimal,ChenG23}, correlation decay \cite{weitz06counting,LinLL14},
and zero-freeness of polynomials \cite{Bar16,patel2017deterministic,GuoLLZ21},
all seem difficult to apply. 
The main obstacle is that these techniques typically make use of properties that hold under arbitrary conditionings.
However, for SFO, even if we start with a graph of minimum degree $3$, 
conditioning the edges can result in a graph that is effectively a cycle, 
in which case no nice property holds.
Our techniques, in contrast, require no hereditary properties and thus can benefit from the minimum degree requirement.

One much less obvious alternative approach to FPTAS is via the connection of the local lemma.
In particular, because SFOs form extremal instances,
their number can be computed via the independence polynomial evaluated at negative weights on the dependency graph. (We also see this fact in \Cref{sec:appendix-lower-bound}.)
Normally this approach would not be efficient, because the dependency graph is usually exponentially large (for example for all-terminal reliability),
but in the case of SFOs, the dependency graph is just the input graph itself.
There are more than one FPTASes \cite{patel2017deterministic,HSV18} for the independence polynomial at negative weights.
However, neither appears able to recover \Cref{thm:sfo-deterministic-counting}.
With the minimum degree $\ge 3$ assumption, the probability vector for SFOs is within the so-called Shearer's region, where both algorithms apply.\footnote{In \cite{patel2017deterministic}, only a uniform bound is stated, but one can introduce a scaling variable $t$ and make a new polynomial in $t$, so that their algorithm works in the Shearer's region.}
The downside is that the running time of both algorithms has the form $(n/\eps)^{O(\log d)}$,\footnote{To be more precise, the hidden constants in the exponents decrease in the multiplicative ``slack'' of how close the evaluated point is to the boundary of Shearer's region. For SFOs, when constant degree vertices are present, the slack is a constant, and so are the hidden constants in the exponents.} where $d$ is the maximum degree of the graph.
Thus, in the setting of \Cref{thm:sfo-deterministic-counting}, these algorithms run in quasi-polynomial time instead.
A more detailed discussion is given in \Cref{sec:ind-poly}.

The rest of the paper is organised as follows.
In \Cref{sec:local-sampler}, we introduce our local sampler.
It is then analysed in \Cref{sec:analysis}.
The main theorems are shown in \Cref{sec:applications}.
We conclude with a few open problems in \Cref{sec:conclusion}.

\section{A local sampler for sink-free orientations}\label{sec:local-sampler}

Fix $G=(V,E)$ as an undirected graph. An orientation $\sigma$ of $G$ is an assignment of a direction to each edge, turning the initial graph into a directed graph. For any $S\subseteq V$, let $\Omega_S$ be the set of $S$-sink-free orientations of $G$, i.e., the set of orientations such that each vertex $v\in S$ is not a sink. Thus, $\Omega_V$ is the set of all (normal) sink-free orientations of $G$. When $\abs{\Omega_S}\neq 0$, we use $\mu_S$ to denote the uniform distribution over $\Omega_S$.
For two adjacent vertices $u,v\in V$, we use $\{u,v\}$ to denote the undirected edge and $(u,v)$ to denote the directed edge, from $u$ to $v$.

We apply the following standard counting-to-sampling reduction \cite{jerrum1986random}. 
Let $V=\{v_1,v_2,\dots,v_n\}$ be arbitrarily ordered and, for each $0\leq i\leq n$, define $V_i=\{v_1,v_2,\dots,v_i\}$. Then, $\abs{\Omega_V}$ can be decomposed into a telescopic product of marginal probabilities:
\begin{equation}\label{eq:decomposition}
\abs{\Omega_V}=\abs{\Omega_{V_0}}\cdot \prod\limits_{i=1}^{n}\frac{\abs{\Omega_{V_i}}}{\abs{\Omega_{V_{i-1}}}}=2^{|E|}\cdot \prod\limits_{i=1}^{n}\mu_{{V_{i-1}}}(v_i\text{ is not a sink}).
\end{equation}
Thus, our goal becomes to estimate $\mu_S(v\text{ is not a sink})$ for any $S\subseteq V$ and $v\not\in S$.

We view $S$-sink-free orientations under the variable framework of the Lov\'{a}sz local lemma.
Here, each edge corresponds to a variable that indicates its direction,
and each vertex in $S$ represents a bad event of being a sink.
An instance is called \emph{extremal} if any two bad events are independent (namely, they share no common variable) or disjoint. 
It is easy to see that all instances to the $S$-sink-free orientation problem are extremal: 
if a vertex is a sink then none of its neighbors can be a sink.
For extremal instances like this,
the celebrated Moser-Tardos algorithm~\cite{moser2010constructive} is guaranteed to output an assignment avoiding all bad events uniformly at random~\cite{GJL19}.
This is summarised in \Cref{Alg:MT}.
Note that when $S=V$, \Cref{Alg:MT} is the sink-popping algorithm by Cohn, Pemantle, and Propp \cite{CPP02}.

%

\begin{algorithm} 
\caption{PRS algorithm for generating an $S$-sink-free orientation} \label{Alg:MT}
\SetKwInOut{Instance}{Instance}
\SetKwInOut{Input}{Input}
\SetKwInOut{Output}{Output}
\SetKwInOut{Data}{Data}
 \Input{
an undirected graph $G=(V,E)$ and a subset of vertices $S\subseteq V$}
\Output{an orientation $\sigma$ of $G$}
orient each edge $e\in E$ uniformly at random and independently to obtain an orientation $\sigma$\;
\While{ $\exists v\in S$ s.t. $v$ is a sink in $\sigma$ }{
    choose such a $v$ arbitrarily\label{Line:choose}\;
    resample the orientation of all edges incident to $v$ in $\sigma$ uniformly at random\;
}
\Return $\sigma$;
\end{algorithm}

The following lemma is a direct corollary from~\cite[Theorem 8]{GJL19} and SFOs being extremal.

\begin{lemma}\label{lemma:MT-SFO-correctness}
If $|\Omega_{S}|\neq 0$, \Cref{Alg:MT} terminates almost surely and returns an orientation distributed exactly as $\mu_S$. 
\end{lemma}

We remark that the only possible case for $\Omega_S=\emptyset$ is when $S$ forms a tree and not connected to any vertex not in $S$.

\Cref{Alg:MT} requires one to generate a global sample when estimating $\mu_S(v\text{ is not a sink})$ for some $v\notin S$, which is wasteful. 
The following observation is crucial to turning it into a local sampler.


\begin{lemma}[criteria for early termination]\label{observation:cycle}
Suppose $|\Omega_S|\neq 0$. For any $v\notin S$, $v$ is a sink upon the termination of \Cref{Alg:MT} if and only if
\begin{enumerate}
    \item[(a)] $v$ becomes a sink at some iteration.
\end{enumerate}
Conversely, $v$ is not a sink upon the termination of \Cref{Alg:MT} if and only if one of the following holds:
\begin{enumerate}
    \item[(b1)] a directed cycle $C$ containing $v$, or a directed path $P$ containing $v$ which ends in a directed cycle $C$ is formed in some iteration, or 
    \item[(b2)] a nonempty directed path $P$ from $v$ to some $u\notin S$ is formed in some iteration. 
\end{enumerate} 
\end{lemma}

\begin{proof}
    First, consider \emph{(a)}. If $v$ becomes a sink at any point, then for every $w \in S$ which is a neighbour of $v$, the edge $(w,v)$ is oriented towards $v$. Since $v \notin S$, the edge $(w,v)$ will not be resampled via $v$, and could only be resampled by $w$ becoming a sink. Since $w$ cannot become a sink without resampling $(w,v)$, $v$ will remain a sink. The other implication is obvious.

    Now for \emph{(b1)} and \emph{(b2)}, we first consider the forward implications. For a cycle $C$, every vertex $u \in C$ has an edge pointing outwards towards some $w \in C$ which also has an edge pointing outwards. None of these edges can be resampled without another edge $e \in C$ being resampled first, so no vertex in the cycle will ever become a sink again.

    Consider a path $P$ that ends outside of $S$ or in a cycle. Inductively, we see that no edge on this path will be resampled without the edge after it being resampled. The edge connected to the cycle or $S^c$ will not be resampled, since that vertex may never be a sink, so no vertex in $P$ can become a sink.

    For the reverse implication, suppose $v$ is eventually not a sink. In that case, there must be some edge $(v,w)$ pointing towards a neighbour $w$. If $w$ is a sink, then $w \notin S$ and we are in case \emph{(b2)}. 
    Otherwise, $w$ is not a sink, and there is an adjacent edge pointing away from $w$, and we repeat this process.
    As the set of vertices is finite, if vertices considered in this process are all in $S$, then there must be repeated vertices eventually,
    in which case we are in \emph{(b1)}.
\end{proof}

\begin{figure}
\begin{tikzpicture}
\begin{scope}[shift={(5.3,0)}]
\node[circle,draw=black,fill=white,pattern=crosshatch] (A0) at (-4.4,-0.4) {};
\node[circle,draw=black,fill=white,pattern=crosshatch] (A1) at (-3.6,1) {};
\node[circle,draw=black,fill=white,pattern=crosshatch] (A2) at (-1.9,0.6) {};
\node[circle,draw=black,fill=white,pattern=crosshatch] (A3) at (-2.2,-0.8) {};
\node[circle,draw=black,fill=white,pattern=crosshatch] (A4) at (-3.6,-1.3) {};
\node[circle,draw=black,fill=white,pattern=crosshatch] (A5) at (-0.8,0.3) {};
\node[circle,draw=black,fill=white] (A6) at (-0.4,-0.9) {};
\node[] () at ($(A6)+(0.1,-0.4)$) {\Large $v$};

\draw[->,red,ultra thick] (A0) to [bend left] (A1);
\draw[->,red,ultra thick] (A1) to [bend left] (A2);
\draw[->,red,ultra thick] (A2) to [bend left] (A3);
\draw[->,red,ultra thick] (A3) to [bend left] (A4);
\draw[->,red,ultra thick] (A4) to [bend left] (A0);
\draw[->,red,ultra thick] (A5) to [bend right] (A2);
\draw[->,red,ultra thick] (A6) to [bend right] (A5);

\draw[->,thick] (A0) to [bend left] ($(A0)+(-0.5,0.2)$);
\draw[<-,thick] (A0) to [bend right] ($(A0)+(0.4,0.4)$);
\draw[->,thick] (A1) to [bend left] ($(A1)+(0.2,0.5)$);
\draw[<-,thick] (A2) to [bend left] ($(A2)+(-0.5,-0.1)$);
\draw[->,thick] (A3) to [bend left] ($(A3)+(0.3,-0.6)$);
\draw[->,thick] (A4) to [bend right] ($(A4)+(-0.4,-0.5)$);
\draw[<-,thick] (A5) to [bend right] ($(A5)+(0.4,0.4)$);
\draw[<-,thick] (A6) to [bend left] ($(A6)+(0.5,0.1)$);
\draw[<-,thick] (A6) to [bend left] ($(A6)+(-0.5,-0.1)$);
\node[] () at (-2.6,-2.3) {(b1)};
\end{scope}

\begin{scope}[shift={(5.6,0)}]
\node[circle,draw=black,fill=white] (B0) at (1.2,0.9) {};
\node[circle,draw=black,fill=white,pattern=crosshatch] (B1) at (2.2,0.3) {};
\node[circle,draw=black,fill=white,pattern=crosshatch] (B2) at (3.2,-0.3) {};
\node[circle,draw=black,fill=white] (B3) at (4.2,-0.9) {};
\node[] () at ($(B3)+(0.1,-0.4)$) {\Large $u$};
\node[] () at ($(B0)+(-0.1,0.4)$) {\Large $v$};

\draw[->,red,ultra thick] (B0) to [bend left] (B1);
\draw[->,red,ultra thick] (B1) to [bend right] (B2);
\draw[->,red,ultra thick] (B2) to [bend left] (B3);
\draw[<-,thick] (B3) to [] ($(B3)+(-0.6,-0.3)$);
\draw[<-,thick] (B3) to [] ($(B3)+(0.6,0.3)$);
\draw[->,thick] (B2) to [] ($(B2)+(0.3,0.6)$);
\draw[<-,thick] (B1) to [] ($(B1)+(-0.6,-0.3)$);
\draw[<-,thick] (B0) to [] ($(B0)+(0.3,0.6)$);
\draw[<-,thick] (B0) to [] ($(B0)+(-0.6,-0.3)$);
\node[] () at (2.6,-2.3) {(b2)};
\end{scope}

\begin{scope}[shift={(-10.2,0)}]
\node[circle,draw=black,fill=white] (C0) at (7.5,0) {};
\node[circle,draw=black,fill=white,pattern=crosshatch] (C1) at (7.2,1.0) {};
\node[circle,draw=black,fill=white,pattern=crosshatch] (C2) at (7,-0.9) {};
\node[circle,draw=black,fill=white] (C3) at (8.6,0.2) {};
\node[] () at ($(C0)+(0.1,-0.4)$) {\Large $v$};
\draw[->,red,ultra thick] (C1) to [] (C0);
\draw[->,red,ultra thick] (C2) to [] (C0);
\draw[->,red,ultra thick] (C3) to [] (C0);
\draw[<-,thick] (C1) to [] ($(C1)+(-0.6,0.2)$);
\draw[<-,thick] (C1) to [] ($(C1)+(0.4,0.4)$);
\draw[<-,thick] (C2) to [] ($(C2)+(-0.6,-0.3)$);
\draw[->,thick] (C2) to [] ($(C2)+(0.2,-0.7)$);
\draw[->,thick] (C3) to [] ($(C3)+(0.6,-0.3)$);
\draw[<-,thick] (C3) to [] ($(C3)+(0.5,0.4)$);
\node[] () at (7.8,-2.3) {(a)};
\end{scope}

\end{tikzpicture}
\caption{Illustration of \Cref{observation:cycle}.
Shaded vertices are in the set $S$. 
Once these patterns are formed, thick red edges would never be resampled in \Cref{Alg:MT}.
}
\label{fig:early-termination}
\end{figure}
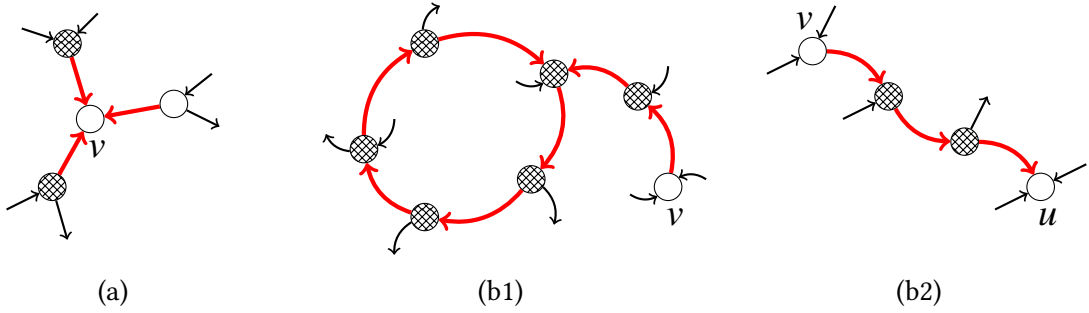


An illustration of \Cref{observation:cycle} is given in \Cref{fig:early-termination}.
Based on \Cref{observation:cycle}, we design a local sampling algorithm for determining whether some vertex $v\notin S$ is a sink or not, given in \Cref{Alg:estimate}. We assume that the undirected graph $G=(V,E)$ is stored as an adjacency list where the neighbors of each vertex are arbitrarily ordered. \Cref{Alg:estimate} takes as input some $S\subseteq V$ and $v\notin S$, returns an indicator variable $x\in \{0,1\}$ such that $\Pr{x=1}=\mu_S(v\text{ is not a sink})$. 
We treat the path $P$ as a subgraph, and $V(P)$ denotes the vertex set of $P$.
When we remove the last vertex from $P$, we remove it and the adjacent edge as well.
Informally, \Cref{Alg:estimate} starts from the vertex $v$,
and reveal adjacent edges one by one.
If there is any edge pointing outward, say $(v,u)$, we move to $u$ and reveal edges adjacent to $u$.
If a sink $w\in S$ is formed, then we mark all adjacent edges of $w$ as unvisited and backtrack.
This induces a directed path starting from $v$.
We repeat this process until any of the early termination criteria of \Cref{observation:cycle} is satisfied,
in which case we halt and output accordingly.

\begin{algorithm} 
\caption{$\sample(G,S,v)$} \label{Alg:estimate}
\SetKwInOut{Instance}{Instance}
\SetKwInOut{Input}{Input}
\SetKwInOut{Output}{Output}
\SetKwInOut{Data}{Data}
\SetKwIF{WP}{ElseIf}{Else}{with probability}{do}{else if}{else}{endif}
 \Input{
an undirected graph $G=(V,E)$, a subset of vertices $S\subseteq V$, and a vertex $v\notin S$;}

\Output{a random value $x\in \{0,1\}$;} 
Let $P$ be a (directed) path initialised as $P=(v)$\;
Initialise a mapping $M:E\to \{\textsf{visited},\textsf{unvisited}\}$ so that $\forall e\in E$, $M(e)=\textsf{unvisited}$\;
\While{$|V(P)|\geq 1$ \label{Line:terminate-sink}}{
    Let $u$ be the last vertex of $P$\label{Line:estimate-choose}\;
    \lIf{$|V(P)|\geq 2$ and $u\notin S$\label{Line:terminate-path}}{\Return $1$}
    \uIf{all edges incident to $u$ are marked \textnormal{\textsf{visited}}\label{Line:sink}}{
        mark all edges incident to $u$ as \textsf{unvisited}\;
        remove $u$ from $P$\label{Line:removal}\; 
    }
    \Else{
        let $e=\{u,w\}$ be the first \textsf{unvisited} edge incident to $u$\;
        mark $e$ as \textsf{visited}\label{Line:mark}\;
        \WP{$1/2$ \label{Line:sample}}{
            \lIf{$w\in V(P)$, or there is a \textnormal{\textsf{visited}} edge $(w,w')$ for some $w'\in V(P)$}{\Return $1$}\label{Line:terminate-cycle-2}
            append $w$ to the end of $P$\label{Line:append}\;
        } 
    }
}
\Return $0$\;
\end{algorithm}

\section{Analysis of the local sampler}
\label{sec:analysis}

In this section, we analyse the correctness and efficiency of \Cref{Alg:estimate}.

\begin{lemma}[correctness of \Cref{Alg:estimate}]\label{lem:alg2-correctness}
Let $G=(V,E)$ be a graph.
For any $S\subseteq V$ such that $|\Omega_S|\neq 0$ and any $v\notin S$, $\sample(G,S,v)$ terminates almost surely and upon termination, returns an $x\in \{0,1\}$ such that 
\[
\Pr{x=1}=\mu_S(v \text{ is not a sink}).
\]
\end{lemma}
\begin{proof}
    We claim that there exists a coupling between the execution of \Cref{Alg:MT} (with input $G,S$ and output $\sigma$), and $\sample(G,S,v)$ (with output $x$), such that 
    \begin{equation}\label{eq:equivalent-condition}
         v\text{ is not a sink under }\sigma\Longleftrightarrow  x=1.
    \end{equation}
    The claim implies the lemma because of \Cref{lemma:MT-SFO-correctness}.

    To prove the claim, we first construct our coupling.
    We use the resampling table idea of Moser and Tardos \cite{moser2010constructive}.
    For each edge, we associate it with an infinite sequence of independent random variables,
    each of which is a uniform orientation.
    This forms the ``resampling table''.
    Our coupling uses the same resampling table for both \Cref{Alg:MT} and \Cref{Alg:estimate}.
    As showed in \cite{moser2010constructive},
    the execution of \Cref{Alg:MT} can be viewed as first constructing this (imaginary) table,
    and whenever the orientation of an edge is randomised, we just reveal and use the next random variable in the sequence.
    For \Cref{Alg:estimate},
    we reveal the orientation of an edge when its status changes from \textsf{unvisited} to \textsf{visited} in \Cref{Line:mark}.
    We execute \Cref{Line:append} if the random orientation from the resampling table is $(u,w)$,
    and otherwise do nothing.
    Namely, in the latter case, the revealed orientation is $(w,u)$, and we just move forward the ``frontier'' of that edge by one step in the resampling table.
    We claim that \eqref{eq:equivalent-condition} holds with this coupling.

    Essentially, the claim holds since for extremal instances, given a resampling table,
    the order of the bad events resampled in \Cref{Alg:MT} does not affect the output.
    This fact is shown in \cite[Section 4]{GJ21}. (See also \cite[Lemma 2.2]{CPP02} for the case of $S=V$.)
    We can ``complete'' \Cref{Alg:estimate} after it finishes. 
    Namely, once \Cref{Alg:estimate} terminates, we randomise all edges that are not yet oriented, and resample edges adjacent to sinks until there are none, using the same resampling table.
    Note that at the end of \Cref{Alg:estimate}, some edges may be marked $\textsf{unvisited}$ but are still oriented.
    Suppose that the output of the completed algorithm is $\sigma'$, an orientation of all edges.
    This completed algorithm is just another implementation of \Cref{Alg:MT} with a specific order of resampling bad events.
    Thus the fact mentioned earlier implies that $\sigma'=\sigma$.
    
    On the other hand, the termination conditions of \Cref{Alg:estimate} correspond to the cases of \Cref{observation:cycle}.
    One can show, via a simple induction over the while loop,
    that at the beginning of the loop,
    the path $P$ is always a directed path starting from $v$,
    and all other visited edges point towards the path $P$.
    This implies that \Cref{Line:terminate-cycle-2} corresponds to case \emph{(b1)} in \Cref{observation:cycle}, 
    \Cref{Line:terminate-path} corresponds to case \emph{(b2)},
    and exiting the while loop in \Cref{Line:terminate-sink} corresponds to case \emph{(a)}.
    When \Cref{Alg:estimate} terminates, we have decided whether or not $v$ is a sink.
    By \Cref{observation:cycle}, this decision stays the same under $\sigma'$.
    As~$\sigma'=\sigma$,~\eqref{eq:equivalent-condition} holds.
\end{proof}

We then analyse the efficiency of \Cref{Alg:estimate}. 
The main bottleneck is when there are degree $2$ vertices.
It would take $\Omega(\ell^2)$ time to resolve an induced path of length $\ell$.
We then focus on the case where the minimum vertex degree is at least $3$. Note that in this case we have $\Omega_S\neq\emptyset$ for any $S\subseteq V$.
For two distributions $P$ and $Q$ over the state space $\Omega$, their total variation distance is defined by $\dtv(P,Q)\defeq\sum_{x\in\Omega}|P(x)-Q(x)|/2$.
For two random variables $x\sim P$ and $y\sim Q$, we also write $\dtv(x,y)$ to denote $\dtv(P,Q)$.
Then we have the following lemma.

\begin{lemma}[efficient truncation of \Cref{Alg:estimate}]\label{lemma:efficient-truncation}
Let $G=(V,E)$ be a graph with minimum degree at least $3$. Let $S\subseteq V$
, $v\notin S$ and $0<\varepsilon<1$. Let $x$ be the output of $\sample(G,S,v)$, and $x'$ constructed as
\begin{itemize}
    \item if $\sample(G,S,v)$ terminates within $72\ln(73/\varepsilon)$ executions of \Cref{Line:sample}, let $x'=x$;
    \item otherwise, let $x'=1$.
\end{itemize}
Then, it holds that
\[
\dtv(x,x')\leq \varepsilon.
\]
\end{lemma}

\begin{proof}
We track the length of the path $P$ during the execution of \Cref{Alg:estimate}. When an edge is chosen in \Cref{Line:estimate-choose} and sampled in \Cref{Line:sample} of \Cref{Alg:estimate}, the following happens:
\begin{itemize}
    \item with probability $1/2$, $w$ is appended to $P$ and the length of $P$ increases by one;
    \item with probability $1/2$, $\{u,w\}$ is marked as \textsf{visited}, and the length of $P$ decreases by one in the next iteration if and only if $\{u,w\}$ was the last \textsf{unvisited} edge incident to $u$.
\end{itemize}
Let $X_i$ be the random variable denoting the length of $P$ after the $i$-th execution of \Cref{Line:sample} in \Cref{Alg:estimate}.
Then the observation above implies that $\{X_i\}_{i\ge 0}$ forms a submartingale. We construct another sequence of random variables $\{Y_i\}_{i\ge 0}$ modified from $\{X_i\}_{i\ge 0}$ as follows:
\begin{itemize}
    \item $Y_0=X_0=1$.
    \item At the $i$-th execution of \Cref{Line:sample} in \Cref{Alg:estimate}:
    \begin{itemize}
        \item if $\{u,w\}$ is the only \textsf{unvisited} edge incident to $u$, set $Y_{i+1}=X_{i+1}-X_{i}+Y_{i}$,
        \item otherwise, set $Y_{i+1}=X_{i+1}-X_{i}+Y_{i}-1/2$.
    \end{itemize}
\end{itemize}
It can be verified that the sequence $\{Y_i\}_{i\ge 0}$ is a martingale. 
\begin{claim}\label{claim:deviation}
    For any $i\geq0$, $X_i-Y_i\geq i/4$.
\end{claim}
\begin{proof}
    Note that $X_i-Y_i=X_{i-1}-Y_{i-1}+c_i$ where $c_i=0$ if $\{u,w\}$ is the only \textsf{unvisited} edge incident to $u$ at the $i$-th execution of \Cref{Line:sample} in \Cref{Alg:estimate}, and $c_i=1/2$ otherwise. 
    Then we can write $X_i-Y_i$ as
    \[X_i-Y_i=X_0-Y_0+\sum_{j=1}^ic_j.\]
    For any $i$ such that $c_i=0$, let $i'$ be the last index such that $c_{i'}=0$, or $i'=0$ if no such $i'$ exists.
    Since the minimum degree of $G$ is at least $3$, when we append any vertex $u$ to $P$, there are at least two unvisited edges incident to $u$.
    It implies that there must be some $j$ such that $i'< j< i$ and $c_j=1/2$. 
    Thus $X_i-Y_i=\sum_{k=1}^ic_k\geq i/4$. 
\end{proof}


Next we show that if \Cref{Alg:estimate} doesn't terminate after $72\ln(73/\varepsilon)$ steps, with high probability the length of the path will not return to $0$.
As $\{Y_i\}_{i\ge 0}$ is a martingale and $\abs{Y_{i+1}-Y_{i}}\leq 3/2$ for all $i\ge 0$, the Azuma–Hoeffding inequality implies that, for any $T>0$ and $C>0$,
\begin{align}\label{eqn:Azuma}
    \Pr{Y_T-Y_0\le -C} \le \exp\left(\frac{C^2}{9T/2}\right).
\end{align}
Thus,
\[
    \Pr{X_{T}=0}\leq \Pr{Y_{T}\leq -T/4}\leq \Pr{Y_{T}-Y_0\leq -T/4}\leq \mathrm{exp}\left(-T/72\right),
\]
where the first inequality is by \Cref{claim:deviation}, and the last inequality is by plugging $C=T/4$ into \eqref{eqn:Azuma}.
Then, we have
\begin{align}\label{eqn:run-time-bound}
\sum\limits_{T=\lceil 72\ln\frac{73}{\varepsilon}\rceil}^{\infty} \Pr{X_{T}=0}\leq \sum\limits_{T=\lceil 72\ln\frac{73}{\varepsilon}\rceil}^{\infty}\mathrm{exp}\left(-T/72\right)\leq \sum\limits_{T= 72\ln\frac{73}{\varepsilon}}^{\infty}\mathrm{exp}\left(-T/72\right)=\frac{\eps}{73(1-\mathrm{e}^{-1/72})}<\varepsilon.    
\end{align}

To finish the proof, we couple $x$ and $x'$ by the same execution of \Cref{Alg:estimate}. 
Thus, if it terminates within $72\ln(73/\varepsilon)$ executions of \Cref{Line:sample},
then $x=x'$ with probability $1$.
If not, \eqref{eqn:run-time-bound} implies that $x=0$ with probability at most $\eps$.
As we always output $x=1$ in this case, $x'\neq x$ with probability at most $\eps$,
which finishes the proof.
\end{proof}

Note that \Cref{lemma:efficient-truncation} does not require $\Omega_S\neq\emptyset$.
This is because it is implied by the minimum degree requirement.
This implication is an easy consequence of the symmetric Shearer's bound.
It is also directly implied by \Cref{lemma:marginal-lower-bound} which we show next.

\section{Applications of the local sampler}
\label{sec:applications}

We show the main theorems in this section.
\Cref{lemma:efficient-truncation} implies an additive error on the truncated estimator.
As we are after relative errors in approximate counting, we need a lower bound of the marginal ratio.

\begin{lemma}\label{lemma:marginal-lower-bound}
   Let $G=(V,E)$ be a graph with a minimum degree at least $3$. For any $S\subseteq V$ and $v\notin S$,
   it holds that $|\Omega_S|\neq 0$ and
    \[
    \mu_S(v \text{ is not a sink})> \frac{1}{2}.
    \]
\end{lemma}

The proof of \Cref{lemma:marginal-lower-bound} can be viewed as an application of the symmetric Shearer's Lemma~\cite{She85} on SFO, and is deferred to \Cref{sec:appendix-lower-bound}. Note that the minimum degree requirement is essential for such a marginal lower bound to hold, as the marginal ratio in \Cref{lemma:marginal-lower-bound} can be of order $O(1/n)$ when $G$ is a cycle and $S=V\setminus \{v\}$.

We then show the two approximate counting algorithms first, namely \Cref{thm:sfo-deterministic-counting} and \Cref{thm:sfo-fast-counting}. 

\begin{proof}[Proof of \Cref{thm:sfo-deterministic-counting}]
By \eqref{eq:decomposition},
we just need to ${\eps/(2n)}$-approximate $\mu_{{V_{i-1}}}(v_i\text{ is not a sink})$ for each $i$ to $\eps$-approximate $|\Omega_V|$, the number of sink-free orientations to $G$. 
The only random choice \Cref{Alg:estimate} makes is \Cref{Line:sample}.
In view of \Cref{lemma:efficient-truncation}, we enumerate the first $72\ln(292n/\eps)$ random choices of choices \Cref{Alg:estimate},
and just output $1$ if the algorithm does not terminate by then.
Let the estimator be the average of all enumeration.
Note that \Cref{lemma:marginal-lower-bound} implies that $\Omega_{V_i}\neq\emptyset$ for any $i$.
Then, \Cref{lem:alg2-correctness,lemma:efficient-truncation} imply that the estimator is an $\eps/(4n)$ additive approximation.
By \Cref{lemma:marginal-lower-bound}, it is also an $\eps/(2n)$ relative approximation,
which is what we need.
    
For the running time, there are $n$ marginals, it takes $\exp(72\ln(292n/\varepsilon))$ enumerations for each marginal probability, 
and each enumeration takes time at most $O(\ln(292n/\varepsilon))$ time. 
Therefore, the overall running time is bounded by $O(n(n/\varepsilon)^{72}\log(n/\varepsilon))$.
\end{proof}

\begin{proof}[Proof of \Cref{thm:sfo-fast-counting}]
We use \eqref{eq:decomposition} again.
Denote $\nu_i=\mu_{V_{i-1}}(v_i\textrm { is not a sink})$ and $\nu=\prod_{i=1}^n\nu_i$.
Let $\widetilde{X}_i\defeq\frac{1}{n}\sum_{i=1}^nx_{i,t}'$ be the average of $n$ independent samples from \Cref{Alg:estimate} truncated after $72\ln(73\times 12n/\eps)$ executions of \Cref{Line:sample}.
Let $\widetilde{X}\defeq\prod_{i=1}^n\widetilde{X}_i$ be an estimator for $\nu$.


For any $i$ and $t$, let $\widetilde{\nu}_i$ be the expectation of $x_{i,t}$ (note that it does not depend on $t$).
By \Cref{lem:alg2-correctness,lemma:efficient-truncation}, 
$\abs{\widetilde{\nu}_i-\nu_i}\le \frac{\eps}{12n}$.
By \Cref{lemma:marginal-lower-bound}, $1-\frac{\eps}{6n}\le\frac{\widetilde{\nu}_i}{\nu_i}\le 1+\frac{\eps}{6n}$.
Let $\widetilde{\nu}=\prod_{i=1}^n \widetilde{\nu}_i$ so that $\E{\widetilde{X}}=\widetilde{\nu}$.
Then, as $0<\eps<1$, 
\begin{align}\label{eqn:exp-error}
    1-\frac{\eps}{3}\le\frac{\widetilde{\nu}}{\nu}\le 1+\frac{\eps}{3}.
\end{align}

We bound $\Var{\widetilde{X}_i}$ and $\Var{\widetilde{X}}$ next.
First,
\begin{align*}
    \Var{\widetilde{X}_i}
    =\Var{\frac{1}{n}\sum_{t=1}^nx_{i,t}'}
    =\frac{1}{n^2}\sum_{t=1}^n\Var{x_{i,t}'}
    \leq\frac{1}{n},
\end{align*}
as each $x_{i,t}'$ is an indicator variable.
Then,
\begin{align*}
    \frac{\Var{\widetilde{X}}}{\left(\E{\widetilde{X}}\right)^2}&=\frac{\E{\widetilde{X}^2}}{\left(\E{\widetilde{X}}\right)^2}-1
    =\frac{\prod_{i=1}^n\E{\widetilde{X}_i^2}}{\prod_{i=1}^n\left(\E{\widetilde{X}_i}\right)^2}-1
    =\prod_{i=1}^n\left(1+\frac{\Var{\widetilde{X}_i}}{\left(\E{\widetilde{X}_i}\right)^2}\right)-1\\
    &\leq\left(1+\frac{4}{n}\right)^n-1<\mathrm{e}^4-1<54. \tag{by \Cref{lemma:marginal-lower-bound}} 
\end{align*} 
To further reduce the variance, let $\widehat{X}$ be the average of $N$ independent samples of $\widetilde{X}$, where $N\defeq\lceil36\times 54/\eps^2\rceil$.
Then, $\mathbf{Var}[\widehat{X}]\leq\frac{\Var{\widetilde{X}}}{N}$.
By Chebyshev's bound, we have
\begin{align*}
    \Pr{\abs{\widehat{X}-\widetilde{\nu}}\geq\frac{\eps}{3}\cdot\widetilde{\nu}}
    \leq\frac{9\mathbf{Var}[\widehat{X}]}{\eps^2\widetilde{\nu}^2}
    \leq\frac{9\times54\eps^2\widetilde{\nu}^2}{36\times 54}\cdot\frac{1}{\eps^2\widetilde{\nu}^2}
    \leq\frac{1}{4}. 
\end{align*}
Thus with probability at least $\frac{3}{4}$, 
we have that $(1-\frac{\varepsilon}{3})\widetilde{\nu}\leq\widehat{X}\leq(1+\frac{\varepsilon}{3})\widetilde{\nu}$. 
By \eqref{eqn:exp-error}, when this holds, $(1-\eps)\nu\leq\widehat{X}\leq(1+\eps)\nu$.
It is then easy to have an $\eps$-approximation of $\abs{\Omega_V}$.

For the running time, each sample $x_{i,t}'$ takes $O(\log(n/\eps))$ time. 
We draw $n$ samples for each of the $n$ vertices, and we repeat this process $N=O(\eps^{-2})$ times. Thus, the total running time is bounded by $O((n/\eps)^2\log(n/\eps))$. 
\end{proof}

For \Cref{thm:sfo-fast-sampling}, we will need a modified version of \Cref{Alg:estimate} to sample from the marginal distributions of the orientation of edges. 
This is given in \Cref{Alg:estimate-edge}.
It takes as input a subset of vertices $S\subseteq V$ and an edge $e\in E$, then outputs a random orientation $\sigma_e$ following the marginal distribution induced from $\mu_S$ on $e$. 
The differences between \Cref{Alg:estimate} and \Cref{Alg:estimate-edge} are:
\begin{itemize}
    \item In \Cref{Alg:estimate}, the number of vertices in $P$ is initialised as $|V(P)|=1$, while in \Cref{Alg:estimate-edge}, it is initialised as $|V(P)|=2$.
    \item When $|V(P)|=1$ and all edges incident to the only vertex $u$ in $P$ are marked as visited:
    \begin{itemize}
        \item In \Cref{Alg:estimate}, the algorithm terminates and returns $0$;
        \item In \Cref{Alg:estimate-edge}, the algorithm terminates if and only if $u\notin S$, and would reinitialise the algorithm otherwise.
    \end{itemize}
\end{itemize}

\begin{algorithm} 
\caption{$\esample(G,S,e)$} \label{Alg:estimate-edge}
\SetKwInOut{Instance}{Instance}
\SetKwInOut{Input}{Input}
\SetKwInOut{Output}{Output}
\SetKwInOut{Data}{Data}
\SetKwIF{WP}{ElseIf}{Else}{with probability}{do}{else if}{else}{endif}
 \Input{
an undirected graph $G=(V,E)$, a subset of vertices $S\subseteq V$, and an edge $e=\{x,y\}\in E$;}

\Output{a random orientation $\sigma_e\in \{(x,y),(y,x)\}$ of $e$;}
Initialise a mapping $M:E\to \{\textsf{visited},\textsf{unvisited}\}$ so that $\forall e\in E$, $M(e)=\textsf{unvisited}$\;
Let $P$ be a (directed) path initialised as $P=(x,y)$ or $P=(y,x)$ with equal probability,
and mark $e$ \textsf{visited}\;
\While{\textnormal{\textsf{True}} \label{Line:edge-terminate-sink}}{
    Let $u$ be the last vertex of $P$\label{Line:edge-estimate-choose}\;
    \lIf{$|V(P)|\geq 2$ and $u\notin S$\label{Line:edge-terminate-path}}{\Return the first edge in $P$}
    \uIf{all edges incident to $u$ are marked \textnormal{\textsf{visited}}\label{Line:edge-sink}}{
        mark all edges incident to $u$ as \textsf{unvisited}\;
        \uIf{$\abs{V(P)}=1$}{
            rerandomise $P$ as $P=(x,y)$ or $P=(y,x)$ with equal probability\;
        }\Else{
            remove $u$ from $P$\label{Line:edge-removal}\;
        }
    }
    \Else{
        let $e=\{u,w\}$ be the first \textsf{unvisited} edge incident to $u$\;
        mark $e$ as \textsf{visited}\;
        \WP{$1/2$ \label{Line:edge-sample}}{
            \If{$w\in V(P)$, or there is a \textnormal{\textsf{visited}} edge $(w,w')$ for some $w'\in V(P)$\label{Line:edge-terminate-cycle}}{\Return the first edge in $P$\;}
            append $w$ to the end of $P$\label{Line:edge-append}\;
        }
    }
}
\end{algorithm}

The correctness of \Cref{Alg:estimate-edge} is due to a coupling argument similar to \Cref{lem:alg2-correctness}.
We couple \Cref{Alg:MT} and \Cref{Alg:estimate-edge} by using the same resampling table.
By the same argument as in \Cref{lem:alg2-correctness}, given the same resampling table,
the orientation of $e$ is the same in the outputs of both \Cref{Alg:MT} and \Cref{Alg:estimate-edge}.
Thus, $\sigma_e$ follows the desired marginal distribution by \Cref{lemma:MT-SFO-correctness}.
As for efficiency, we notice that the same martingale argument as in \Cref{lemma:efficient-truncation} applies to the length of $P$ as well.
Early truncation of the edge sampler only incurs a small error. 
However, we need some extra care for the self-reduction in the overall sampling algorithm.

\begin{proof}[Proof of \Cref{thm:sfo-fast-sampling}]
We sequentially sample the orientation of edges in $G$ (approximately) from its conditional marginal distribution. 
Suppose we choose an edge $e=\{u,v\}$,
and the sampled orientation is $(u,v)$.
Then, we can remove $e$ from the graph, and let $S\gets S\setminus\{u\}$.
The conditional distribution is effectively the same as $\mu_S$ in the remaining graph.

One subtlety here is that doing so may create vertices of degree $\le 2$.
To cope with this, we keep sampling edges adjacent to one vertex in $S$ as much as possible before moving on to the next.
Suppose the current focus is on $v$.
We use $\esample$ to sample the orientation of edges adjacent to $v$ one at a time until either $v$ is removed from $S$ or the degree of $v$ becomes $1$.
In the latter case, the leftover edge must be oriented away from $v$, which also results in removing $v$ from $S$.
Note that, when either condition holds, the last edge sampled is oriented as $(v,u)$ for some neighbour $u$ of $v$.
We then move our focus to $u$ if $u\in S$, and move to an arbitrary vertex in $S$ otherwise.
The key property of choosing edges this way is that, whenever $\esample(G,S,e)$ is invoked, there can only be at most one vertex of degree $2$ in $S$,
and if it exists, it must be an endpoint of $e$.
If all vertices are removed from $S$, we finish by simply outputing a uniformly at random orientation of the remaining edges.

To maintain efficiency, we truncate $\esample(G,S,e)$ in each step of the sampling process.
More specifically, for some constant $C$, we output the first edge of $P$ once the number of executions of \Cref{Line:edge-sample} in \Cref{Alg:estimate-edge} exceeds $C\ln(m/\eps)$.
We claim that there is a constant $C$ such that the truncation only incurs an $\eps/m$ error in total variation distance between the output and the marginal distribution.
This is because the same martingale argument as in \Cref{lemma:efficient-truncation} still applies.
Note that if $P$ visits any vertex not in $S$, the algorithm immediately terminates.
Thus degrees of vertices not in $S$ do not affect the argument.
Moreover, the only degree $2$ vertex in $S$, say $x$, is adjacent to the first edge $e=\{x,y\}$ of $P$.
If $e$ is initialised as $\{x,y\}$, then when $P$ returns to $x$, the algorithm immediately terminates.
Otherwise $e$ is initialised as $\{y,x\}$, in which case there is no drift in the first step of $P$.
Thus, as long as we adjust the constant to compensate the potential lack of drift in the first step,
the martingale argument in \Cref{lemma:efficient-truncation} still works and the claim holds.
As the truncation error is $\eps/m$, we may couple the untruncated algorithm with the truncated version,
and a union bound implies that the overall error is at most $\eps$.

As we process each edge in at most $O(\log(m/\eps))$ time,
the overall running time is then $O(m\log(m/\eps))$.
This finishes the proof of the fast sampling algorithm.
%
\end{proof}

\subsection{Proof of the marginal lower bound}\label{sec:appendix-lower-bound}

Now we prove the lower bound of marginal ratios for SFOs, namely, \Cref{lemma:marginal-lower-bound}.
Let us first recall the variable framework of the local lemma.
Consider the probability space $\+P$ of a uniformly random orientation of $G$ (namely orienting each edge independently and uniformly at random), and each $u\in S$ corresponding to a bad event $\+E_u$ of $u$ being a sink.
We then have 
\[
p_u\defeq \Pr[\+P]{\+E_u}=2^{-d(u)}, \quad \forall u\in V,
\]
where $d(u)$ denotes the degree of $u$.
We also need some definitions, essentially from~\cite{HV17} and small variations from those in~\cite{She85}.

\begin{definition}
    We define the following notations.
\begin{itemize}
    \item Let $\Ind(G)$ denote all independent sets of $G$, i.e.,
    \[
    \Ind(G)\defeq \{I\subseteq V\mid I\text{ contains no edge of }G\}.
    \]
    \item For $J\subseteq V$, let
    \begin{align}\label{eqn:ind-poly}
        q_J\defeq \sum\limits_{\substack{I\in \Ind(G)\\ I\subseteq J}}(-1)^{\abs{I}}\prod\limits_{u\in I}p_u,
    \end{align}
    and
    \[
    P_J\defeq \Pr[\+P]{\bigwedge\limits_{u\in J}\neg \+E_u},
    \]
    which is the probability under $\+P$ that all vertices in $J$ are sink-free. 
\end{itemize}
\end{definition}
We then proceed to the proof.

\begin{proof}[Proof of \texorpdfstring{\Cref{lemma:marginal-lower-bound}}{Lemma 4.1}]
For $u\in V$, let $\Gamma(u)$ denote the set of neighbours of $u$ in $G$.
We claim that for any $J\subseteq V$ and $u\in J$:
\begin{enumerate}
    \item $P_J=q_J>0$;\label{item:shearer-induction-1}
    \item $\frac{q_J}{q_{J\setminus \{u\}}}>\begin{cases}
        \frac{1}{2} & \Gamma(u)\subseteq J;\\
        \frac{2}{3} & \text{otherwise}.
    \end{cases}$\label{item:shearer-induction-3}
\end{enumerate}
\Cref{lemma:marginal-lower-bound} immediately follows from the claim as $P_S=\frac{\abs{\Omega_S}}{2^{\abs{E}}}>0$ and 
\[
\mu_S(v \text{ is not a sink})=\Pr[\+P]{\neg \+E_v\mid \bigwedge\limits_{u\in S}\neg \+E_u}=\frac{P_{S\cup \{v\}}}{P_S} = \frac{q_{S\cup \{v\}}}{q_S}>\frac{1}{2},
\]
where the last equality is by \Cref{item:shearer-induction-1} and the inequality is by \Cref{item:shearer-induction-3}. 

We then prove claim by induction on the size of $J$. The base case is when $J=\emptyset$, and all items directly hold as $P_{\emptyset}=q_{\emptyset}=1$.
For the induction step, 
we first prove \Cref{item:shearer-induction-1}. 
For $J\subseteq V$ and $u\in J$,
denote $\Gamma^+(u)=\Gamma(u)\cup \{u\}.$ We have
\begin{align*}
  \Pr[\+P]{\+E_u\land\left( \bigwedge\limits_{j\in J\setminus \{u\}}\neg \+E_j\right)} & = \Pr[\+P]{\+E_u} \Pr[\+P]{\left( \bigwedge\limits_{j\in J\setminus \{u\}}\neg \+E_j\right)\mid \+E_u}
  = p_u \Pr[\+P]{\left( \bigwedge\limits_{j\in J\setminus \Gamma^+(u)}\neg \+E_j\right)\mid \+E_u}\\
   & = p_u \Pr[\+P]{\left( \bigwedge\limits_{j\in J\setminus \Gamma^+(u)}\neg \+E_j\right)} = p_u\cdot P_{J\setminus \Gamma^+(u)},
\end{align*}
where in the second equality we used the fact that $S$-SFO instances are extremal.
Thus,
\begin{align}\label{eq:marginal-lower-bound-1}
    P_J&=P_{J\setminus \{u\}}-\Pr[\+P]{\+E_u\land\left( \bigwedge\limits_{j\in J\setminus \{u\}}\neg \+E_j\right)} = P_{J\setminus \{u\}}-p_u\cdot P_{J\setminus \Gamma^+(u)}.
\end{align}
Also, by separating independent sets according to whether they contain $u$ or not, we have
\begin{align}
    q_J=&\sum\limits_{\substack{I\in \Ind(G)\\ I\subseteq J}}(-1)^{\abs{I}}\prod\limits_{i\in I}p_i \notag\\
    =&\sum\limits_{\substack{I\in \Ind(G)\\ I\subseteq J\setminus \{u\}}}(-1)^{\abs{I}}\prod\limits_{i\in I}p_i-p_u\cdot \sum\limits_{\substack{I\in \Ind(G)\\ I\subseteq J\setminus \Gamma^+(u)}}(-1)^{\abs{I}}\prod\limits_{i\in I}p_i\notag\\
    =&q_{J\setminus \{u\}}-p_u\cdot q_{J\setminus \Gamma^+(u)}.\label{eq:marginal-lower-bound-2}
\end{align}
Combining \eqref{eq:marginal-lower-bound-1}, \eqref{eq:marginal-lower-bound-2}, and the induction hypothesis,
we have that $P_J=q_J$.
For the positivity,
let $J\cap \Gamma^+(u)$ be listed as $\{u,u_1,\dots,u_k\}$ for some $k\le d(u)$.
For $0\leq i\leq k$, let $U_i=\{u,u_1,\dots,u_i\}$.
Then we have
\begin{align*}
    \frac{q_{J\setminus \{u\}}}{q_{J\setminus \Gamma^+(u)}} 
    = \prod_{i=1}^k \frac{q_{J\setminus U_{i-1}}}{q_{(J\setminus U_{i-1})\setminus \{u_i\}}}
    > 2^{-k} \ge 2^{-d(u)} = p_u,
\end{align*}
where the first inequality is by \Cref{item:shearer-induction-3} of the induction hypothesis.
Thus, $q_J>0$ by \eqref{eq:marginal-lower-bound-2} and \Cref{item:shearer-induction-1} holds.


It remains to show \Cref{item:shearer-induction-3}. 
Recall that $J\cap \Gamma^+(u)$ is listed as $\{u,u_1,\dots,u_k\}$. 
For any $0\le i\le k$, $\Gamma^+(u_i)\not\subseteq J\setminus U_i$ because $u\in \Gamma^+(u_i)$ and $u\notin (J\setminus U_i)$.
Then by the induction hypothesis on the second case of \Cref{item:shearer-induction-3},
\begin{align}\label{eqn:q_u_i}
    \frac{q_{(J\setminus U_i)\setminus \{u_i\}}}{q_{J\setminus U_{i}}} < \frac{3}{2}.
\end{align}
By \eqref{eq:marginal-lower-bound-2}, we have
\begin{equation}\label{eq:marginal-lower-bound-3}
\frac{q_J}{q_{J\setminus \{u\}}}=1-p_u\cdot \frac{q_{J\setminus \Gamma^+(u)}}{q_{J\setminus \{u\}}}=1-2^{-d(u)}\cdot \prod\limits_{i=0}^{k-1}\frac{q_{(J\setminus U_i)\setminus \{u_i\}}}{q_{J\setminus U_{i}}} \stackrel{\eqref{eqn:q_u_i}}{>} 1-2^{-d(u)}\cdot \left(\frac{3}{2}\right)^{k}.
\end{equation}
If $\Gamma(u)\subseteq J$, $k\le d(u)$ and as $d(u)\ge 3$,
\begin{align*}
    1-2^{-d(u)}\cdot \left(\frac{3}{2}\right)^{k} \ge 1- \left(\frac{3}{4}\right)^{d(u)} \ge \frac{37}{64} > \frac{1}{2}.
\end{align*}
If $\Gamma(u)\not \subseteq J$, $k\le d(u)-1$ and, again, as $d(u)\ge 3$,
\begin{align*}
    1-2^{-d(u)}\cdot \left(\frac{3}{2}\right)^{k} \ge 1- \frac{2}{3}\cdot\left(\frac{3}{4}\right)^{d(u)} \ge \frac{23}{32} > \frac{2}{3}.
\end{align*}
Together with \eqref{eq:marginal-lower-bound-3},
this finishes the proof of \Cref{item:shearer-induction-3} and the lemma.
\end{proof}

In the proof above, \Cref{item:shearer-induction-1} holds mainly because the instance is extremal.
For general non-extremal cases, we would have $\frac{P_J}{P_{J\setminus \{u\}}}\geq \frac{q_J}{q_{J\setminus \{u\}}}$ for $J\subseteq V$ and $u\in J$ instead.

\subsection{Independence polynomial at negative weights}\label{sec:ind-poly}

An interesting consequence of \Cref{item:shearer-induction-1} in the proof of \Cref{lemma:marginal-lower-bound} is that the number of SFOs can be computed using the independence polynomial evaluated at negative activities.
More specifically, similar to \eqref{eqn:ind-poly}, let 
\begin{align*}
    q_G(\mathbf{x}) = \sum\limits_{I\in \Ind(G)}\prod\limits_{u\in I}x_u,
\end{align*}
where $\mathbf{x}$ is a vector of weights for each vertex.
Then, $q_G(-\mathbf{p})=q_V$ where $q_V$ is defined in \eqref{eqn:ind-poly},
and thus $\abs{\Omega_V} = 2^{\abs{E}}q_G(-\mathbf{p})$, where $\mathbf{p}$ is the vector $(p_u)_{u\in V}$ of failure probabilities at the vertices.
Namely $p_u=2^{-d(u)}$ where $d(u)$ is the degree of $u$.

There are more than one FPTASes \cite{patel2017deterministic,HSV18} that can efficiently approximate the independence polynomial at negative weights.
%
These algorithms work in the so-called Shearer's region \cite{She85}.
To explain Shearer's region, let us abuse the notation slightly and extend the definition in \eqref{eqn:ind-poly} to a function $q_J(\mathbf{x})=\sum\limits_{I\in \Ind(G),I\subseteq J}\prod\limits_{u\in I}x_u$ to take an input weight vector $\mathbf{x}$.
Then, a vector $\mathbf{p}$ is in Shearer's region if and only if $q_S(-\mathbf{p})>0$ for all $S\subseteq V$.
\Cref{lemma:marginal-lower-bound} implies that the probability vector for SFOs is in Shearer's region.
Moreover, we say a vector $\mathbf{p}$ has slack $\alpha$ if $(1+\alpha)\mathbf{p}$ is in Shearer's region.
For a vector $\mathbf{x}$ with slack $\alpha$,
the algorithm by Patel and Regts \cite{patel2017deterministic} $\eps$-approximates $q_G(\mathbf{x})$ in time $(n/\eps)^{O(\log d/\alpha)}$, 
and the algorithm by Harvey, Srivastava, and Vondr\'{a}k \cite{HSV18} runs in time $(n/(\alpha \eps))^{O(\log d/\sqrt{\alpha})}$,
where $d$ is the maximum degree of the graph.
They do not recover \Cref{thm:sfo-deterministic-counting} as the slack is a constant when constant degree vertices exist.
If, in the meantime, some other vertices have unbounded degrees, these algorithms run in quasi-polynomial time instead.

To see the last point, we construct a graph that contains vertices of unbounded degrees but with constant slack for SFOs.
Consider the wheel graph $G$, which consists of a cycle $C_n$ of length $n$, and a central vertex $v$ that connects to all vertices of $C_n$.
Thus, $p_v=2^{-n}$ and $p_u=1/8$ for any $u$ in $C_n$.
For the cycle, as there are two SFOs, we see that $q_{C_n}(-\mathbf{1/4})=2^{-n+1}$ (where we use \Cref{item:shearer-induction-1} in the proof of \Cref{lemma:marginal-lower-bound}).
Thus, by \eqref{eq:marginal-lower-bound-2},
\begin{align*}
    q_{G}(-2\mathbf{p}) &= q_{C_n}(-\mathbf{1/4}) - 2p_v = 2^{-n+1} - 2\cdot 2^{-n}=0.
\end{align*}
Therefore, the slack here is at most $1$, despite the existence of a high degree vertex.

In summary, the existing FPTASes on the independence polynomial with negative weights do not handle the mixture of high and low degree vertices well enough for the case of SFOs.
However, it might provide an alternative approach to derive FPTASes to count solutions to extremal instances of the local lemma, which is worthy of further study.

\section{Concluding remarks}
\label{sec:conclusion}

Originally, Bubley and Dyer \cite{BD97a} introduced sink-free orientations as a special case of read-twice \textsc{Sat}.
Here, ``read-twice'' means that each variable in a CNF formula appears exactly twice, and it corresponds to an edge of the graph.
Vertices of the graph correspond to clauses of the formula.
The assignment of the edge is an orientation.
This represents that the variable appears with opposite signs in the formula.
In fact, Bubley and Dyer showed an FPRAS for all read-twice \#\textsc{Sat}.
It is natural to ask if they admit FPTAS as well.
This question was first raised by Lin, Liu, and Lu \cite{LinLL14},
who also gave an FPTAS for monotone read-twice \#\textsc{Sat} (which is equivalent to counting edge-covers in graphs).
The monotone requirement means that the two appearances of any variable have the same sign.
From this perspective, our FPTAS is complementary to that of \cite{LinLL14}.
However, as our techniques are drastically different from \cite{LinLL14},
to give an FPTAS for all read-twice \#\textsc{Sat}, one may need to find a way to combine these two techniques to handle mixed appearances of variables.

Another immediate question is to generalise our local sampler under the partial rejection sampling framework.
The first step would be to be able to handle degree 2 vertices for SFOs,
which breaks our current submartingale argument.
To go a bit further, a local sampler for all extremal instances would yield an FPTAS for all-terminal reliability, whose existence is a major open problem.
Also, for all-terminal reliability, one may also attempt to localise the near-linear time sampler in \cite{CGZZ24}.

Lastly, in addition to the discussion of \Cref{sec:ind-poly},
let us discuss another polynomial associated with SFOs and its zero-freeness.
Fix a SFO $\sigma$.
Let $p(x)=\sum_{i=0}^m C_i x^i$, where $m$ is the number of edges, and $C_i$ indicates how many SFOs exist that agree with $\sigma$ in exactly $i$ edges.
It is easy to evaluate $p(0)=1$, and $p(1)$ is the total number of SFOs.
However, for a cycle, this polynomial becomes $1+x^m$, which can have a zero arbitrarily close to $1$.
This zero defeats, at least, the standard application of Barvinok's method \cite{Bar16,patel2017deterministic}.
Although one could exclude cycles by requiring the minimum degree to be at least $3$ (like we did in this paper), current techniques of proving zero-freeness seem to hinge on handling all subgraphs.
For example, to use Ruelle's contraction like in \cite{GuoLLZ21}, one has to start from small fragments of the graph and gradually rebuild it. 
The obstacle then is to avoid starting from or encountering cycles in the rebuilding process.
Other methods, such as the recursion-based one in \cite{LSS19}, require hereditary properties (similar to the so-called strong spatial mixing) that break in cycles as well.
It would be interesting to see if any of our arguments can help in proving zero-freeness of the polynomial above.

\section*{Acknowledgement}
We thank Guus Regts for helpful comments on an earlier version of this paper.

\bibliographystyle{alpha}
\bibliography{references} 
\clearpage
\end{document}